\documentclass[12pt]{article}
\usepackage[utf8]{inputenc}

\title{Stochastic Continuum Models for High--Entropy Alloys with Short-range Order}
\author{Yahong Yang$^{a}$\thanks{E-mail address: \it\textbf{yyangct@connect.ust.hk}},
Luchan Zhang$^{b}$\thanks{E-mail address: \it\textbf{zhanglc@szu.edu.cn}},
Yang Xiang$^{a,c}$\thanks{E-mail address: \it\textbf{maxiang@ust.hk}}\\$^{a}$\small\textit{Department of Mathematics, Hong Kong University of Science and Technology,}\\ \small\textit{Clear Water Bay, Kowloon, Hong Kong}\\
$^{b}$\small\textit{College of Mathematics and Statistics, Shenzhen University,} \\ \small\textit{Shenzhen 518060, China}\\
$^{c}$\small\textit{HKUST Shenzhen-Hong Kong
Collaborative Innovation Research Institute,} \\ \small\textit{Futian, Shenzhen, China}}

\usepackage[numbers]{natbib}
\usepackage{amsmath}
\usepackage{epsfig,amsthm}
\usepackage{amssymb,amsfonts}
\usepackage{dsfont}
\usepackage{subfigure}
\usepackage{indentfirst}
\usepackage{mathrsfs}
\usepackage{color}
\usepackage{graphicx}
\usepackage{zymacros}
\usepackage{hyperref}

\begin{document}
	
	\maketitle
	
	\begin{abstract}
High entropy alloys (HEAs) are a class of novel materials that exhibit superb engineering properties. It has been demonstrated by extensive experiments and first principles/atomistic simulations that short-range order in the atomic level randomness strongly influences the properties of HEAs. In this paper, we derive stochastic continuum models for HEAs with short-range order from atomistic models.
A proper continuum limit is obtained such that the mean and variance of the atomic level randomness together with the
short-range order described by a characteristic length are kept in the process from the atomistic interaction model to the continuum equation.
The obtained continuum model with short-range order is in the form of an Ornstein--Uhlenbeck (OU) process.   This validates the continuum  model based on the OU process adopted phenomenologically by Zhang et al. [Acta Mater., 166 (2019), pp. 424--434] for  HEAs with short-range order.
We derive such stochastic continuum models with short-range order for both  elasticity in HEAs without defects and  HEAs with dislocations (line defects). The obtained stochastic continuum models are based on the energy formulations, whose variations lead to stochastic partial differential equations.
	\end{abstract}
	
	
	\section{Introduction}
	
	High entropy alloys (HEAs) are single phase crystals with random solid solutions of five or more elements of nearly equal composition \cite{yeh2004nanostructured,cantor2004microstructural,zhang2014microstructures,miracle2017critical}.
It is widely believed that HEAs have many ideal engineering properties, such as high strength, high temperature stability,  high fracture resistance, etc. Therefore, HEAs have attracted considerable research interest in the development of advance materials
\cite{yeh2004nanostructured,cantor2004microstructural,senkov2011microstructure,del2012modeling,zhang2014microstructures,toda2015modelling,singh2015atomic,
toda2015interatomic,tamm2015atomic,sharma2016atomistic,varvenne2016theory,varvenne2017solute,miracle2017critical,feng2017effects,
fernandez2017short,ZhangFX2017,zhang2019effect,ikeda2019ab,George2020,zhang2020short,WangYZ2020,ZhangYongwei2021,wu2021short}.
 Although HEAs and their applications have been widely investigated in the materials science area, mathematical understandings and rigorous developments of models to describe HEAs are still limited.

The strength of HEAs, as in the traditional crystalline materials, is associated with the motion of dislocations (line defects) driven by the stress.
There are HEA models based on independent randomness of the element types at individual atomic sites, i.e., without short-range order. Models for the strength of HEAs \cite{senkov2011microstructure,toda2015modelling,varvenne2016theory,varvenne2017solute,George2020} have been developed that generalize the solute solution strengthening model for traditional alloys \cite{labusch1970statistical}. In the models proposed in  Ref.~\cite{varvenne2016theory,varvenne2017solute,George2020}, dislocations interact with the HEA lattice through the long-range elastic field, and the elastic field in the HEA is modeled by considering each lattice site as a point defect with random perturbation in its size.  Strength of HEAs influenced by the dislocation core effect in considered in the continuum model in Ref.~\cite{zhang2019effect} by a stochastic generalization of the Peierls--Nabarro model \cite{peierls1940size,nabarro1947dislocations}.
Recently, Jiang {\it et al.}~\cite{jiang2020stochastic} presented a mathematical derivation of the stochastic continuum model proposed phenomenologically in Ref.~\cite{zhang2019effect} from an atomistic model for dislocations in bilayer HEAs, by using asymptotic analysis and limit theorems; short-range order was not considered in this derivation.


At finite temperature, it has been shown by experiments and atomistic simulations that the distributions of elements are commonly not completely random in HEAs:  the element type at an atomic site will enhance or reduce the probability of element types around it, i.e., the correlation between the element types at two close atomic sites is not $0$. This is the short-range order in HEAs. The Warren-Cowley pair-correlation parameters \cite{cowley1995diffraction} is one of the widely used classical methods to describe short-range order in
muliti-component systems including HEAs \cite{singh2015atomic,tamm2015atomic,sharma2016atomistic,feng2017effects,fernandez2017short,ZhangYongwei2021}:
    	$\alpha_{e_ie_j}(r)=1-\frac{P_{e_ie_j}(r)}{p_j}$,
    where $P_{e_ie_j}(r)$ is the probability of finding an atom of type $e_j$ at site $j$ given an atom of type $e_i$ at site $i$,   $p_j$ is the probability of element $e_j$ at site $j$, and $r$ is the distance between the atomic sites $i$ and $j$.
It has been shown by first principles calculations and atomistic simulations \cite{del2012modeling,singh2015atomic,toda2015interatomic,tamm2015atomic,sharma2016atomistic,feng2017effects,fernandez2017short,
ikeda2019ab,ZhangYongwei2021} and experiments \cite{ZhangFX2017,zhang2020short,wu2021short} that short-range order strongly influences the properties of HEAs.

Despite the active research on  HEAs with short-range order as reviewed above, almost all the models for the short-range order in HEAs are atomistic models or first principles calculations on even smaller scales. The only available continuum model is the one proposed by Zhang {\it et al.} \cite{zhang2019effect}, in which the short-range order in HEAs is incorporated by the Ornstein--Uhlenbeck (OU) process \cite{uhlenbeck1930theory,doob1942brownian,wang1945theory,karatzas1998brownian} in the continuum stochastic Peierls-Nabarro model for dislocations in HEAs. This model predicts significant increase of the intrinsic strength of HEAs as the correlation length or the standard deviation of the randomness in the HEAs increases, which  is
consistent with experimental measurements of  the yield strength of  HEAs \cite{senkov2011microstructure}. This model was proposed phenomenologically, and no derivation from atomistic model is available for  continuum level description of the short-range order in HEAs.

In this paper, we derive stochastic continuum models for HEAs with short-range order from atomistic models.
Unlike the derivation presented in \cite{jiang2020stochastic} for the continuum model of HEAs without short-range order, the challenge here is how to define and obtain the continuum limit of the atomic-level randomness with short-range correlations.
For this purpose, we first identify a characteristic length $H$ of the short-range order in the atomistic model. Under the assumptions of fast decaying nature of the short-range order and that the characteristic length $H$ of the short-range order is much larger than the lattice  constant but is much smaller than the length scale of the continuum model, we obtain the continuum limit from the atomistic interactions with short-range order. A proper continuum limit is defined such that the short-range order is kept in the process from the atomistic model to the continuum equation. The obtained continuum model with short-range order is in the form of an OU process, which validates the HEA model adopted phenomenologically in \cite{zhang2019effect}. We derive such stochastic continuum models with short-range order for both (i) the elastic deformation in HEAs without defects and (ii) HEAs with dislocations (in the form of the Peierls-Nabarro model). The obtained stochastic continuum models are based on the energy formulation. We briefly discuss the variational formulation of these obtained stochastic energies at the end of this paper.

	\section{Stochastic Elasticity Model for HEAs with Short Range Order}\label{elas}

	\subsection{Atomistic model of one row of atoms without defects}\label{subsec:atom_single}
	In this subsection, we establish the atomistic model of HEAs in one row without defects. In an HEA, each atom site is randomly occupied by one of the main elements. Denote the set of all these elements by $\Omega$:
	\begin{equation}
		\Omega=\{e_1,~e_2,\cdots,~e_m\}
	\end{equation}
For each lattice site, a random variable $\omega$ is defined with sample space $\Omega$ and probability measure:
\begin{equation}
		\rmP(e_k):=p_k\ge 0, \ \ \ \sum_{i=1}^{m}p_i=1.\label{simple}
	\end{equation}
Here $p_k$ is the probability of element $e_k$ occupying the lattice site.
We denote  $\omega_j$	to be such a random variable at the $j$-th site in the HEA. The location of the $j$-th atom is denoted by $a_j$. See Fig.~\ref{one row} for an illustration.

	\begin{figure}[h!]
		\centering
		\includegraphics[scale=0.97]{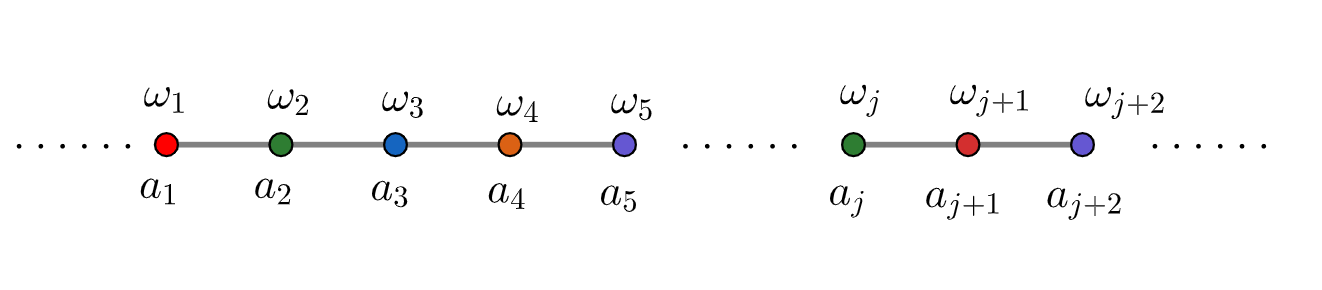}
		\caption{An HEA  of one row of atoms, where $a_j$ is the location  and $\omega_j$	is the random variable at the $j$-th atom site.}
		\label{one row}
	\end{figure}
	
	Consider an elastic displacement field $\{u_j\}_{j\in \sZ}\subset\sR$ on the HEA lattice $\{a_j\}_{j\in \sZ}$. Suppose that the HEA system is described by  pairwise potentials. Without loss of generality, we consider the nearest neighbor interaction. In the HEA system, the pairwise potentials depend on not only  the distance between the two atoms but also their elements. Therefore, the interaction energy between atoms $a_j$ and $a_{j+1}$ can be written as
	\begin{equation}
		V\left(h_j+u_{j+1}-u_j,\omega_j,\omega_{j+1}\right),
	\end{equation}
	where $h_j:=h(\omega_j,\omega_{j+1})$ is the random lattice constant which is the solution of \begin{equation}
		\frac{\,\D V(r,\omega_j,\omega_{j+1})}{\,\D r}\Big|_{r=h_j}=0.\label{random lattice constant}
	\end{equation}
Here $r$ is the distance between these two atoms.
    Hence the total energy of the HEA using the atomistic model is
    \begin{equation}
    	E_\text{a-el}=\sum_{j\in \sZ}\left[V\left(h_j+u_{j+1}-u_j,\omega_j,\omega_{j+1}\right)-V\left(h_j,\omega_j,\omega_{j+1}\right)\right].\label{total atom energy}
    \end{equation}

This elastic energy has the following approximate formula:
   \begin{align}
   	E_\text{a-el}=&\sum_{j\in \sZ}\left[V\left(h_j+u_{j+1}-u_j,\omega_j,\omega_{j+1}\right)-V\left(h_j,\omega_j,\omega_{j+1}\right)\right]\notag\\
   \approx&\frac{1}{2}\sum_{j\in \sZ}V''_jh_j^2\left(\frac{u_{j+1}-u_{j}}{h_j}\right)^2\notag\\
   =& \frac{1}{2}\sum_{j\in \sZ} \beta_j \left(\frac{u_{j+1}-u_{j}}{h_j}\right)^2h_j,\label{atom_constant}
   \end{align}
   where
   \begin{flalign}
   \beta_j:=&V''_jh_j,\label{atom_constant_beta}\\
   	V''_j:=&\frac{\,\D^2 V(r,\omega_j,\omega_{j+1})}{\,\D r^2}\Big|_{r=h_j}.\label{atom_constant_beta1}
   \end{flalign}
Here $\beta_j$ can be considered as the elastic modulus on the atomic scale. In fact, from Eq.~\eqref{atom_constant}, we have 	$E_\text{a-el}\approx \int_{x\in\sR}\frac{1}{2}\beta_\text{atom}(x)\left(\frac{\,\D u}{\,\D x}\right)^2\,\D x$, where
$\beta_\text{atom}(x)\approx\beta_j$ for $x\in [a_j,a_{j+1})$.

\begin{rmk}\label{review elasticz}
  Note that in the classical elasticity theory in one-dimension in which there is no randomness, the elastic energy associated with the displacement $u$ is
  \begin{equation}
    	\int_{\sR}\frac{1}{2}{\beta}\left(\frac{\,\D u}{\,\D x}\right)^2\,\D x,
    \end{equation}
  where ${\beta}$ is the elastic modulus.

   This elastic energy can be formally obtained by the corresponding deterministic atomistic model with pair potential $V$ as:
    \begin{equation}
    	\sum_{j\in \sZ}\left[V\left(h+u_{j+1}-u_j\right)-V\left(h\right)\right]\approx \int_{\sR}\frac{1}{2}V''(h)h\left(\frac{\,\D u}{\,\D x}\right)^2\,\D x,
    \end{equation}
where $h$ is the lattice constant.  The elastic modulus is $\beta=V''(h)h$.

In this case, the equilibrium equation without body forces is
  \begin{equation}
    \beta\frac{\,\D^2 u}{\,\D x^2}=0.
  \end{equation}

  \end{rmk}

 \subsection{Assumptions for short-range order in atomistic model of HEAs and limit theorems}\label{assum elas}
   In this subsections, we present the atomistic model and assumptions for the short-range order  in HEAs.
We employ the $\alpha$-mixing coefficients $\alpha_n$ \cite{billingsley2008probability}  to describe the short-range order in HEAs.
For the random variable sequence $\{X_j\}_{j\in\sZ}$,
 the $\alpha$-mixing coefficient $\alpha_n$ is defined as \cite{billingsley2008probability}
 \begin{equation}
   			\alpha_{n}=\sup \left\{|P(A \cap B)-P(A) P(B)|:  A \in \mathcal{F}_{-\infty}^{k}, B \in \mathcal{F}_{k+n}^{+\infty}, \ \forall k\in \sZ\right\}\label{beta}
   		\end{equation}
   where $\mathcal{F}_{a}^{b}$ is the $\sigma$-field generated by  $\{X_a,~X_{a+1},~\cdots,~X_b\}$.

    Recall that the method of Warren-Cowley pair-correlation parameters \cite{cowley1995diffraction} is one of the widely used classical methods to describe short-range order in
muliti-component systems including HEAs (e.g., \cite{singh2015atomic,tamm2015atomic,sharma2016atomistic,feng2017effects,fernandez2017short,ZhangYongwei2021}):
   	$\alpha_{e_ie_j}(n):=1-\frac{P_{e_ie_j}(n)}{p_j}$,
    where $P_{e_ie_j}(n)$ is the probability of finding an atom of type $e_j$ at $a_{n+r}$ given an atom of type $e_i$ at $a_{r}$, and $p_j$ is the probability of element $e_j$ occupying the lattice site defined in Eq.~\eqref{simple}.
    The	$\alpha$-mixing coefficients $\alpha_n$  of $\{\omega_j\}_{j\in\sZ}$ in HEAs are stronger than the pair-correlation parameters 	$\alpha_{e_ie_j}(n)$
     since correlations between groups of atoms are also considered in the definition of the	$\alpha$-mixing coefficients:
   	\begin{align*}\alpha_{n}=&\sup \left\{|P(A \cap B)-P(A) P(B)|: A \in \mathcal{F}_{-\infty}^{k}, B \in \mathcal{F}_{k+n}^{+\infty}, \ \forall k\in \sZ\right\}\notag\\
   \ge&\sup_k \big\{|\rmP(\omega_{k}=e_i,\omega_{k+n}=e_j)-\rmP(\omega_{k}=e_i)\rmP(\omega_{k+n}=e_j)|\big\}, \  e_i, e_j\in\Omega \notag\\
   =& p_ip_j|\alpha_{e_ie_j}(n)|\\
   \ge &\lambda|\alpha_{e_ie_j}(n)|,
   	\end{align*}
   	where $\lambda:=\min\{p_i^2\}$.

  In this paper, we consider a short range order in HEAs that is rapidly decaying with atomic distance as in the following assumption:
   \begin{assump}\label{assump1}
   	There is a constant number $N_s$ and a constant $C$, such that the	$\alpha$-mixing coefficients $\alpha_n$  of $\{\omega_j\}_{j\in\sZ}$ in HEAs satisfies
   \begin{equation}
   		\alpha_n\le C{n}^{-5} \ {\rm for} \  n\leq N_s, \ {\rm and}\  \alpha_n=0 \ {\rm for} \ n>N_s.\label{eqn:decayrate}
   	\end{equation}
   This means that $\omega_j$ and $\omega_{j+n}$ are independent if $n>N_s$.
   \end{assump}

 The rapidly decaying short-range order described in Assumption \ref{assump1} is consistent with the property of short-range order in alloys and HEAs. In Refs.~\cite{cowley1950approximate,sethna2006statistical}, they showed that short-range order exists in alloys when the temperature  is greater than the critical temperature, and it decays quickly with the atomic distance. In Refs.~\cite{feng2017effects,fernandez2017short,tamm2015atomic}, they only considered the first and second nearest-neighbor shell short range orders ($\alpha_{e_ie_j}(n)$ for $n=1,2$) in HEAs based on the fact that the correlation parameter $\alpha_{e_ie_j}(n)\to 0$ quickly as $n$ increases and the first and second nearest-neighbor shell short range orders play  dominant role in the correlation effect.
Moreover, the rapid decay of the $\alpha$-mixing coefficients $\alpha_n$ implies rapid decay of the correlation. This is can be proved by a lemma in Ref.~\cite{billingsley2008probability} (Lemma 2 on Page 365) that the correlation is bounded by the $\alpha$-mixing coefficient.
	


We will use the following generalized central limit theorem in our derivation.  Note that it holds for $\{\omega_j\}_{j\in\sZ}$ due to the Assumption \ref{assump1}.

   {\bf Theorem}~\cite[Theorem 27.4]{billingsley2008probability}
   \emph{Suppose that random variables $X_{1}, X_{2}, \cdots$ are stationary with $\alpha$-mixing coefficient $\alpha_{n}=O\left(n^{-5}\right),$ and $\rmE\left(X_{n}\right)=0$, $\rmE\left[X_{n}^{12}\right]<\infty$. Let $S_{n}=X_{1}+$
   	$\cdots+X_{n}$ and $\sigma^{2}={\displaystyle \lim _{n \rightarrow \infty} \rmE\left[S_{n}^{2}\right]/n}$, where $\sigma$ is positive. Then
   	\begin{equation}
   	\frac{S_{n}}{\sigma \sqrt{n}} \stackrel{d}\longrightarrow \mathcal{N}(0,1), \quad \text { as } n \rightarrow \infty,
   	\end{equation}
   	where $\mathcal{N}(0,1)$ is the standard Gaussian distribution, and $\stackrel{d}\longrightarrow$ is the convergence in distribution.
}
%

As described in the previous subsection, we consider the nearest neighbor interaction in this paper. The following lemma is able to give the relationship between the  $\alpha$-mixing coefficient of  $\{\omega_j\}$ and the  $\alpha$-mixing coefficient of the pairwise interaction energies associated with $\{(\omega_j,\omega_{j+1})\}$.
\begin{lem}\label{barH}
 Let $\alpha_n(\mathbf X)$  be the $\alpha$-mixing coefficient  of the random variable sequence  $\{X_j\}_{j\in\sZ}$. If for another random variable sequence  $\{Y_j\}_{j\in\sZ}$, each $Y_j$ is a Borel measurable function of $X_j$ and $X_{j+1}$, i.e., $Y_j=f_j(X_j,X_{j+1})$ for some Borel measurable $f_j$,
 and let $\alpha_n(\mathbf Y)$  be the $\alpha$-mixing coefficient  of the random variable sequence  $\{Y_j\}_{j\in\sZ}$, then we have  for any $n\geq 1$,
\begin{equation}
 \alpha_n(\mathbf Y)\leq \alpha_{n-1}(\mathbf X).
			\end{equation}
		\end{lem}
This lemma can be proved by the definition of the  $\alpha$-mixing coefficient, and the fact that $\sigma(Y_{j})\subset\sigma(X_j,X_{j+1})$ and as a result, $\sigma(Y_k,Y_{k+1},\cdots,Y_{k+n})\subset\sigma(X_k,X_{k+1},\cdots,X_{k+n+1})$.


Based on  Assumption \ref{assump1}, we define the following length to characterize the range of short-range order, based on the nearest neighbor interaction energies associated with $\{(\omega_j,\omega_{j+1})\}$ (i.e., the bounds between atoms $\{(a_j,a_{j+1})\}$).

\begin{defi}[Length of range of short-range order]\label{def1}
Denote the length of range of nonzero short range order with respect to an atom as $H$:
\begin{equation}
   	H:=n_s\bar{h}, \ n_s:=2N_s+2, \ {\rm where} \ \bar{h}=\rmE(h_j).\label{smallest range}
   \end{equation}
\end{defi}

 \begin{figure}[htbp]
   \centering
   \includegraphics[scale=0.97]{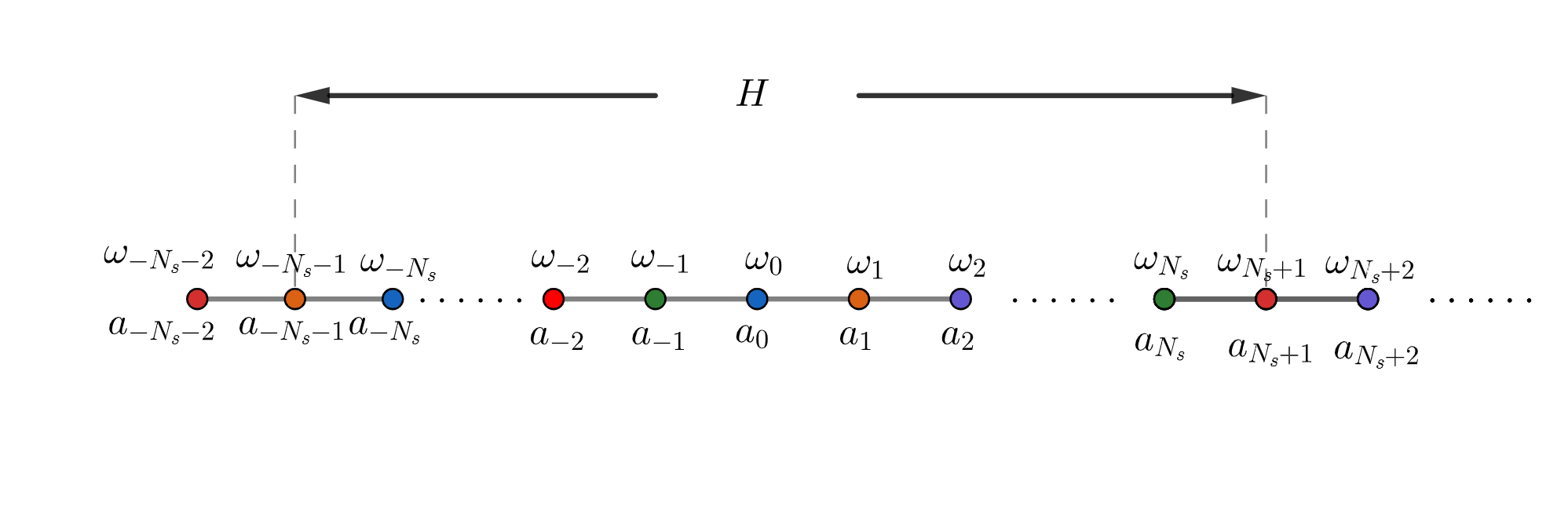}
   \caption{The range of nonzero short range order for the bond between atoms $a_0$ and $a_1$ (i.e., with respect to the random variable pair $(\omega_0,\omega_{1})$).}
   \label{h}
   \end{figure}

 Fig.~\ref{h} illustrates this length $H$ of nonzero short range order for the bond between atoms $a_0$ and $a_1$, i.e., it is independent with those bonds outside this range. Note that $H/2$ is the correlation length in this atomistic model.
	Here we neglect the perturbation of lattice constant for simplicity in this definition. Our analysis also works for stochastic $H$.

  In order to derive a continuum model from the atomistic model, we assume that there are three length scales: the atomic scale, the supercell with size of $H$, and the continuum scale. This is summarized in the following assumption.
  \begin{assump}\label{atom1}
   Let	 $h_j$ for $j\in \sZ$ be the random lattice constant given by Eq.~(\ref{random lattice constant}) with $\bar{h}=\rmE(h_j)$, and $H$ be the length of short-range order defined in Eq.~(\ref{smallest range}). Then we have
   \begin{equation}
   \bar{h}\ll H \ll L,
   \end{equation}
   where $L$ is the length scale of the continuum model. Moreover, we assume that the factor $C$ in Eq.~\eqref{eqn:decayrate} for the decay of the $\alpha$-mixing coefficient $\alpha_n$ satisfies $C=O(1)$.
  \end{assump}

Finally in this subsection, we define the continuum limit of the average of random variables, which will be used  to obtain the stochastic continuum models.

\begin{defi}[Continuum limit of average of random variables]\label{def:average}
  For the  sequence of random variables $\{X_i\}_{i=1}^{n}$, 
  we want to approximate each random variable by the average of $\{X_i\}_{i=1}^{n}$ to obtain the approximate sequence $\{Y_i\}_{i=1}^{n}$.
  Assume that  for each $X_i$, it can be written as:
  \begin{equation}
  		X_i=\rmE(X_i)+\lambda_i U+R_i,
  	\end{equation}
  where $\lambda_i$ is a deterministic number, $U$ is a random variable with $\rmE(U)=0$, and $\{R_i\}_{i=1}^{n}$ are a sequence of random variables with  $\alpha$-mixing coefficient $\alpha_m=O(m^{-5})$. If the following limits exist:
  \begin{flalign}
  E=&\lim_{n \to +\infty}\frac{1}{n}\sum_{i=1}^{n}\rmE(X_i),\\
   \lambda=&\lim_{n \to +\infty}\frac{1}{n}\sum_{i=1}^{n}\lambda_i,\\
   \Delta_R^2=&\lim_{n\to+\infty}\frac{1}{n}\operatorname{Var}\left(\sum_{i=1}^{n}R_i\right),
  \end{flalign}
   then as $n\to\infty$, the averaged random variable $Y_i$ is defined as
  \begin{equation}
   	Y_i:=E+\lambda U+Z_i,\label{average variable}
  \end{equation}
  where $\{Z_i\}_{i=1}^{n}$ are a sequence of independent identically distributed (i.i.d.) random variables with the Gaussian distribution $\fN(0,\Delta_R^2)$.
  \end{defi}

  This definition of continuum limit of average of random variables is different from that in the deterministic case. In addition to the continuum limit of the simple average $A={\displaystyle \lim_{n\to\infty}}\frac{1}{n}\sum_{i=1}^{n}X_i=E+\lambda U$, there is also a contribution ${\displaystyle \lim_{n\to\infty}}\frac{1}{\sqrt{n}}\sum_{i=1}^{n}R_i=Z_i$ to account for the average of variances, which comes from the weakly dependent sequence $\{R_i\}_{i=1}^{n}$ and vanishes in the simple average: ${\displaystyle \lim_{n\to\infty}}\frac{1}{n}\sum_{i=1}^{n}R_i=0$ due to the generalized central limit theorem  {\bf Theorem}~\cite[Theorem 27.4]{billingsley2008probability} shown above. The above definition guarantees that $ \sum_{i=1}^n Z_i$ and $\sum_{i=1}^n R_i$ have the same collective behavior, i.e., $\sim \fN(0,n\Delta_R^2)$ for large $n$.

  For an example, for the i.i.d. random variables $\{X_i\}_{i=1}^{n}$ with Gaussian distribution $\fN(E,\sigma^2)$, the averaged sequence given by the above definition are $\{Y_i\}_{i=1}^{n}$ that are still i.i.d. with Gaussian distribution $\fN(E,\sigma^2)$, which is the desired result.
%
  On the other hand, if we simply use ${\displaystyle \lim_{n\to\infty}}\frac{1}{n}\sum_{i=1}^{n}X_i$ to approximate each random variable,  then we have $Y_i= E$ for all $i$. The information of variance, i.e., the information of randomness, is lost in this continuum limit.

  \subsection{Stochastic elasticity model}\label{sto elas}
  In this subsection, we derive a continuum stochastic elasticity theory in HEAs with short range order from the  atomistic model. We start from continuum approximation of the atomic-level stochastic elastic modulus $\{\beta_j\}$ defined in Eq.~\eqref{atom_constant_beta}. From its definition,  $\beta_j=\beta(\omega_j, \omega_{j+1})$ depending on the interaction of the bond between atoms $a_j$ and $a_{j+1}$. As discussed in the previous subsection on the interaction energies associated with atomic bonds,
  these $\beta_j$'s have identical distribution by the setting of the system as described in the previous subsection, but they are not necessarily independent due to the existence of the short-range order.
   The length of range of short-range order of $\{\beta_j\}$ is  $H$ specified in Definition \ref{def1}.
   We also have the following two lemmas that are related to the short-range order of $\{\beta_j\}$.


   \begin{lem} \label{lemma1}
   For the atomic-level stochastic elastic modulus $\{\beta_j\}$ defined in Eqs.~\eqref{atom_constant_beta} and \eqref{atom_constant_beta1},
 we have
   \begin{equation}
   		\alpha_n\le C{n}^{-5} \ {\rm for} \  n\leq N_s+1, \ {\rm and}\  \alpha_n=0 \ {\rm for} \ n>N_s+1,\label{eqn:decayrate-beta}
   	\end{equation}
   and this means that $\beta_j$ and $\beta_{j+n}$ are independent if $n>N_s+1$, i.e.,
 \begin{equation}
  			\operatorname{Cov}(\beta_j,\beta_{j+n})=0, \ \text{ for } \  n>N_s+1.
  		\end{equation}
  \end{lem}

The conclusions in this lemma come directly from Assumption \ref{assump1},  the definition of $\{\beta_j\}$ in Eqs.~\eqref{atom_constant_beta} and \eqref{atom_constant_beta1},
and Lemma \ref{barH}.

%

\begin{lem}\label{convergence}
	Suppose Assumption \ref{assump1} holds. Consider the average of variance:
\begin{equation}
		f(m):=\frac{1}{m}\operatorname{Var}\left(\sum_{j=1}^{m}\beta_j\right)=\frac{1}{m}\sum_{1\le i,j\le m}\operatorname{Cov}(\beta_j,\beta_i).
	\end{equation}
We have
\begin{equation}
\lim_{m\to+\infty}f(m)=\Delta_e^2,
\end{equation}
where
\begin{equation}
\Delta_e^2:=\sum_{j= -N_s-1}^{N_s+1}\operatorname{Cov}(\beta_j,\beta_0).\label{variance}
	\end{equation}
Furthermore,  there is a constant $C$, such that for $m\ge 2N_s+3$,
\begin{equation}
	\left|f(m)-\Delta_e^2\right|\leq\frac{C}{m}.
\end{equation}
\end{lem}
\begin{proof}
  For $m\ge 2N_s+3$, we have
  	\begin{align} 		
  \left|f(m)-\Delta_e^2\right|
  \le
&\left|\frac{1}{m}\sum_{k=N_s+2}^{m-N_s-1}\left[\sum_{j=1}^{m}\operatorname{Cov}(\beta_j,\beta_k)-\Delta_e^2\right]
 \right|\notag\\
  &+\left|\frac{1}{m}\left(\sum_{k=1}^{N_s+1}+\sum_{k=m-N_s}^{m}\right)
  \left[\sum_{j=1}^{m}\operatorname{Cov}(\beta_j,\beta_k)-\Delta_e^2\right]\right|\notag\\
  \le&0+\frac{2N_s+2}{m}(v^*+\Delta_e^2)\notag\\
  =&\frac{2N_s+2}{m}(v^*+\Delta_e^2), \label{converge}
  	\end{align}
  	where $v^*=\sum_{j= -N_s-1}^{N_s+1}\left|\operatorname{Cov}(\beta_j,\beta_0)\right|$. Here in the first inequality, the summation with respect to $k$ in $f(m)$ is divided into three parts: $\sum_{k=N_s+2}^{m-N_s-1}$, $\sum_{k=1}^{N_s+1}$ and $\sum_{k=m-N_s}^{m}$.
  Note that $\sum_{j=1}^{m}\operatorname{Cov}(\beta_j,\beta_k)=\sum_{j=k-N_s-1}^{k+N_s+1}\operatorname{Cov}(\beta_j,\beta_k)=\Delta_e^2$ in the first part.
   Denote $C:=(2N_s+2)(v^*+\Delta_e^2)$, then we have $\left|f(m)-\Delta_e^2\right|\le C/m$ when $m\ge 2N_{s}+3$. Hence ${\displaystyle \lim_{m\to+\infty}}f(m)=\Delta_e^2$ holds.
  \end{proof}

Now we consider the continuum approximation of $\D \beta(x)$.  Starting from $\beta_i$, consider $\{\beta_j\}$ over a cell with size $H$ defined in Eq.~\eqref{smallest range} which contains $n_s$ atomic sites, i.e.,
 $\beta_{k+i}$, $k=0, 1, 2,\cdots, n_s$. We have
  \begin{equation}
   	\sum_{k=0}^{n_s}(\beta_{k+i}-\beta_{i})=(n_s+1)\bar{\beta}-(n_s+1)\beta_i+\sum_{k=0}^{n_s}\left(\beta_{k+i}-\rmE \beta_{k+i}\right),\label{atom equation 2}
  \end{equation}
where
  \begin{equation}
  	\bar{\beta}:=\rmE\beta_i.
  \end{equation}
Note that $\bar{\beta}=\rmE\beta_j$ for any $j$.

   Let $\beta(x)$ be the elastic modulus on the continuum scale. We want to use the average of $ \sum_{k=0}^{n_s}(\beta_{k+i}-\beta_{i})$ to approximate $\frac{1}{2}\,\D \beta(x)$ over the length $H$.
   (Note that in the deterministic case, $\sum_{k=0}^{n_s}(\beta_{k+i}-\beta_{i})\approx\frac{1}{2}n_s(n_s+1)\beta'(a_i)\bar{h}$, while $\frac{1}{2}\,\D \beta(x)\approx\frac{1}{2}n_s\beta'(a_i)\bar{h}$, i.e., $\frac{1}{n_s+1}\sum_{k=0}^{n_s}(\beta_{k+i}-\beta_{i})\approx\frac{1}{2}\,\D \beta(x)$.) From Eq.~\eqref{atom equation 2},  the average of $ \sum_{k=0}^{n_s}(\beta_{k+i}-\beta_{i})$ equals the average of  $(n_s+1)\bar{\beta}-(n_s+1)\beta_i+\sum_{k=0}^{n_s}\left(\beta_{k+i}-\rmE \beta_{k+i}\right)$,
   which by Definition~\ref{def:average} of the continuum limit of average of random variables and the decay property of the $\alpha$-mixing coefficient $\alpha_n$ of $\{\beta_j\}$ in Lemma~\ref{lemma1}, is approximately
  $\bar{\beta}-\beta_i+\mathcal{N}(0,\Delta_e^2)$,  where $\Delta_e^2$ is defined in Eq.~\eqref{variance}.
    Note that
the limit $n_s\to \infty$ holds due to the condition $H\gg\bar{h}$ in Assumption~\ref{atom1}, where $n_s=H/\bar{h}$.

To summarize,  the continuum limit of the average of $\sum_{k=0}^{n_s}(\beta_{k+i}-\beta_{i})$, i.e., $\frac{1}{2}\,\D \beta(x)$ over length $H$ starting from site $a_i$, is
  \begin{equation}\label{limit equation00}
   	\frac{1}{2}\,\D \beta(x)\approx\bar{\beta}-\beta_i+\mathcal{N}(0,\Delta_e^2).
   \end{equation}
   This averaged increment defines $\beta_{i+\frac{n_s}{2}}$ in the continuum formulation in the middle of the supercell between locations $a_i$ (with value $\beta_{i}$) and
   $a_{i+n_s}$ (with value $\beta_{i+n_s}$). That is, in the continuum model,  the values of  $\beta_{i+jn_s}$, for integer $j$, are inherited directly from the atomistic model, and the values of $\beta_{i+(j+\frac{1}{2})n_s}$ are defined through the averaged increment within a supercell in the atomistic model as given in Eq.~\eqref{limit equation00}; see an illustration of the values defined in the continuum model in Fig.~\ref{link}.  The averaged increment  $\frac{1}{2}\,\D \beta$ obtained in
    Eq.~\eqref{limit equation00} serves as a link from atomistic model to continuum model that passes the atomic level short-range order to the continuum model.


\begin{figure}[htbp]
	\centering
	\includegraphics[width=0.9\linewidth]{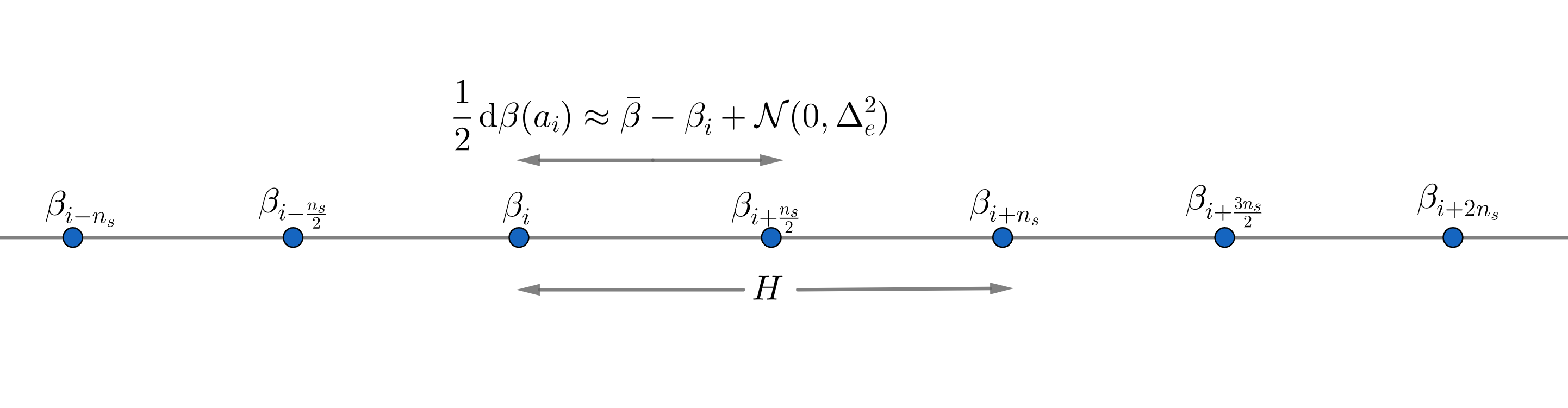}
	\caption{ Illustration of the values of $\beta(x)$ defined in the continuum model. In the continuum model,  the values of  $\beta_{i+jn_s}$, for integer $j$, are inherited directly from the atomistic model; and the values of $\beta_{i+(j+\frac{1}{2})n_s}$ are defined through the averaged increment within a supercell in the atomistic model as given in Eq.~\eqref{limit equation00}, through which the atomic level short-range order is passed to the continuum model.}
	\label{link}
\end{figure}

Since we are approximating $\,\D \beta$ over the length of $H$ scale, setting $\,\D x=H$ and $\,\D B_x=B_{x+H}-B_{x}$, where $B_x$ is the Brownian motion, we have the approximation at $x=a_i$:
   \begin{equation}
    \frac{1}{2}\,\D \beta(x)=\frac{\bar{\beta}-\beta(x)}{H}\,\D x+\frac{\Delta_e}{\sqrt{H}}\,\D B_x.
    \label{limit equation}
   \end{equation}
  This is the Ornstein--Uhlenbeck (OU) process \cite{uhlenbeck1930theory,doob1942brownian,wang1945theory,karatzas1998brownian}. The OU process  has been employed phenomenologically   in Ref.~\cite{zhang2019effect} to model the short-range order in HEAs, and here we provide a rigorous derivation from stochastic atomistic model with parameters directly from the atomistic model.

The continuum model in Eq.~(\ref{limit equation}) was obtained based on the approximation of $\D \beta(x)$ over a supercell with size $H$, which is defined in  Eq.~\eqref{smallest range}. This continuum formulation strongly depends on $H$, and here we discuss more on why this length is appropriate for the continuum approximation of $\D \beta(x)$. First, since we want to have a well-defined continuum limit, the length $H_\beta$ over which $\D \beta(x)$ is obtained has to be no less than $H$, so that the variance $\Delta_e^2$ defined in \eqref{variance} does not change when $H_\beta$ is further increased. Moreover, on the continuum level, $\beta(x)$ satisfies Eq.~\eqref{limit equation}, which is an OU process whose correlation length is $H/2$ (given in Eq.~\eqref{cov} below). This correlation length in the continuum model agrees with that in the atomistic model. On the other hand, if $\D \beta(x)$ is approximated by the average over a length $H_\beta>H$, we will have parameter $H_\beta$ instead of $H$ in the continuum model in Eq.~\eqref{limit equation}, and as a result, the correlation length in the continuum model will be $H_\beta/2$, which is strictly greater than the correlation length $H/2$ in the atomistic model. Therefore, $H$ is the appropriate length for the continuum approximation of $\D \beta(x)$ from atomistic model.


From Eq.~(\ref{limit equation}) and the solution formula of the OU process \cite{uhlenbeck1930theory,doob1942brownian,wang1945theory,karatzas1998brownian},
the stochastic elastic modulus $\beta(x)$  is
   \begin{equation}
   	\beta(x)=\bar{\beta}+\Delta_eY_x,\label{solution}
   \end{equation}
   where $\{Y_x\}_{x\in\sR}$ is an OU process:
   \begin{equation}
   Y_x=\int_{-\infty}^x\frac{2}{\sqrt{H}}e^{\frac{2}{H}(s-x)}\,\D B_s.\label{ou_solution}
   \end{equation}
For each point $x$, $\beta(x)$ is a Gaussian:
\begin{equation}
   	\beta(x)\sim\mathcal{N}\left(\bar{\beta},\Delta_e^2\right), \label{distribution}
   \end{equation}
   and the correlation of $\beta(x)$ at two points $x_1$ and $x_2$ is
  	\begin{equation}
  		\operatorname{Cov}(\beta(x_1),\beta(x_2))=\Delta_e^2e^{-\frac{2|x_1-x_2|}{H}}.\label{cov}
  	\end{equation}
Therefore, $\Delta_e$ indicates the amplitude of randomness at each lattice site, and $H/2$ is the correlation length. These agree with the definitions of $\Delta_e^2$ in \eqref{variance} (which is the average of variances at $n_s$ lattice sites) and $H$ in \eqref{smallest range} in the atomistic model.

  With this continuum stochastic elastic modulus, the stochastic elastic energy $E_\text{elas}$ satisfies
   \begin{equation}\label{elas-energy000}
   	\,\D E_\text{elas}=\frac{1}{2}\beta(x)\left(\frac{\,\D u}{\,\D x}\right)^2\,\D x.
   \end{equation}

%
%


 When the range of the short-range order, i.e., the correlation length, $H\to+\infty$,   by Eq.~\eqref{cov},  $\rmE\left([\beta(x_1)-\beta(x_2)]^2\right)=2\Delta_e^2\left(1-e^{-\frac{2|x_1-x_2|}{H}}\right)\to0$ for any $x_1$ and $x_2$, which is the case of uniform randomness. When $H\to 0$, by Eq.~\eqref{cov},  $\operatorname{Cov}(\beta(x_1),\beta(x_2))\to 0$ for $x_1\neq x_2$, which means that $\beta(x_1)$ and $\beta(x_2)$ are independent for any $x_1\neq x_2$. Especially, in the regime of $H\to 0$ with $\beta(x_1)$ and $\beta(x_2)$ being independent for different lattice sites $x_1\neq x_2$, the right-hand side of Eq.~\eqref{limit equation} dominates, and Eq.~\eqref{limit equation} becomes
 $\beta(x)\,\D x= \bar{\beta}\,\D x+\Delta_e\sqrt{\bar{h}}\,\D B_x$, where $\bar{h}$ is the average lattice constant which is the smallest distance between atomic sites. In this independent case, we have
  \begin{equation}
   	\,\D E_\text{elas}=\frac{1}{2} \left(\frac{\,\D u}{\,\D x}\right)^2\left(\bar{\beta}\,\D x+\Delta_e\sqrt{\bar{h}}\,\D B_x\right).\label{indpendent elas}
   \end{equation}
   This agrees with the energy formulation derived in Ref.~\cite{jiang2020stochastic} ((5.14) there) under the assumption of independent randomness, i.e., without short-range order.

%
%

In the continuum formulation of the elastic energy 	$E_\text{elas}$ given in  Eq.~\eqref{elas-energy000}, both the elastic modulus $\beta(x)$ and the displacement gradient $\frac{\,\D u}{\,\D x}$ contain the effect of randomness. When the effect of randomness is small, we have
    \begin{flalign}
   	E_\text{elas}=&\int_\sR\frac{1}{2}\beta(x)\left(\frac{\,\D u}{\,\D x}\right)^2\,\D x\notag\\=&\int_\sR\frac{1}{2}(\bar{\beta}+\delta\beta(x))\left(\frac{\,\D \bar{u}}{\,\D x}+\frac{\,\D\delta u}{\,\D x}\right)^2\,\D x\,\notag\\
   \approx&\int_\sR\frac{1}{2}\bar{\beta}\left(\frac{\,\D \bar{u}}{\,\D x}\right)^2\,\D x+\int_\sR\bar{\beta}\frac{\,\D \bar{u}}{\,\D x}\frac{\,\D \delta u}{\,\D x}\,\D x+\int_\sR\frac{1}{2}\delta\beta(x)\left(\frac{\,\D \bar{u}}{\,\D x}\right)^2\,\D x,\label{point defect}
   \end{flalign}
    where $\bar{u}(x)=\rmE u(x)$.
    Here the first term  is the $O(1)$ elastic energy, the second term is the leading order perturbation due to the randomness in displacement (i.e., randomness in lattice constant), and
    the third term is leading order perturbation due to the randomness effect in elastic modulus.

   \section{Stochastic Continuum Model for Dislocations in HEAs with Short-range Order}\label{PN}

   In this section, we consider derivation of continuum model for defects in HEAs with short-range order. We focus on dislocations that are line defects in crystalline materials  \cite{Hirth1982}.

   \subsection{Review of classical Peierls--Nabarro model for dislocations}\label{subsec:cPN}
   The Peierls-Nabarro models \cite{,peierls1940size,nabarro1947dislocations,vitek1968intrinsic,Hirth1982,Xiang2006} are continuum models for dislocations that incorporate the atomistic structure.
   In the classical Peierls--Nabarro model for a dislocation in a crystal with a single type atoms, the slip plane of the dislocation separates the entire system into two continuums described by linear elasticity theory, and the interaction across the slip plane is modeled by a nonlinear potential (the $\gamma$-surface) that comes from the atomic interaction \cite{vitek1968intrinsic}.
   This continuum description enables nonlinear interaction within the dislocation core where the atomic structure is heavily distorted.
   In the Peierls-Nabarro model for the two-layer system, 
   each layer is a continuum governed by linear elasticity, and  there is a nonlinear interaction between the two layers; see Fig.~\ref{class}(b) for an illustration of the atomic structure.

   Consider an edge dislocation in a bilayer system as illustrated in Fig.~\ref{class}(b). The Burger vector of the dislocation is $\vb=(\bar{h},0)$. The displacements in the $x$ direction (along the layer) of the top and bottom layers are $u^+(x)$ and $u^-(x)$, respectively.
  The disregistry (relative displacement) between the two layers is
   \begin{equation}
    	\phi(x):=u^+(x)-u^-(x).
    \end{equation}
For this edge dislocation, $\phi(x)$ satisfies    the boundary condition
     \begin{equation}
    	\lim_{x\to-\infty}\phi(x)=0,~\lim_{x\to+\infty}\phi(x)=\bar{h},
    \end{equation}
where $\bar{h}$ is the lattice constant; see  Fig.~\ref{class}(a).

     \begin{figure}[htbp]
   	\centering
 \includegraphics[width=0.8\textwidth]{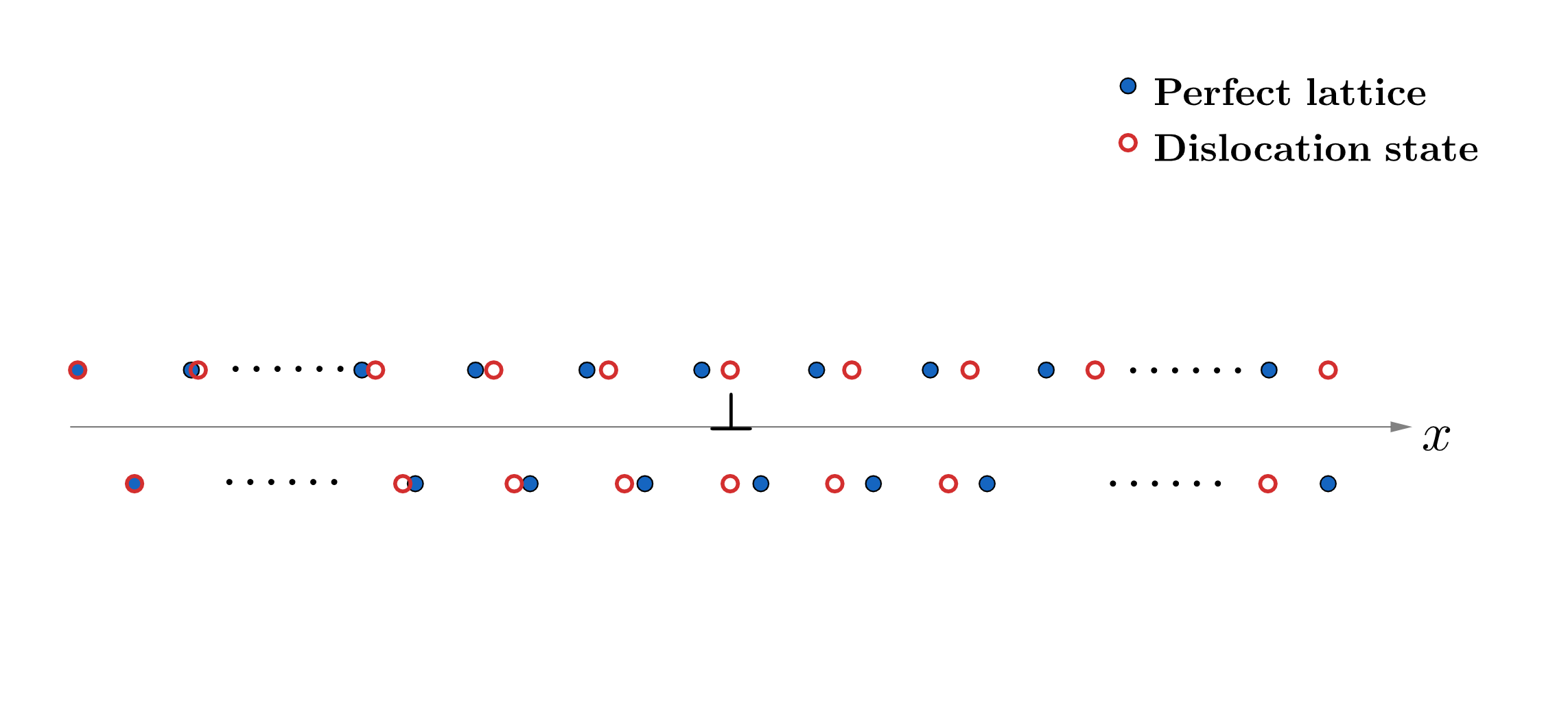}
   	\caption{An edge dislocation in a bilayer system. (a) Peierls-Nabarro model: Schematic
   		illustration of disregistry $\phi(x)$. The sharp transition region is the core of the dislocation. (b) Atomistic model: Schematic
   		illustration of locations of atoms associated with this edge dislocation. The perfect lattice is a triangular lattice. The notation $\perp$ shows location of the dislocation.}
   	\label{class}
   \end{figure}

  In the Peierls-Nabarro model, the total energy is written as the sum of the elastic energy and the misfit energy. The elastic energy is the energy in the two continuums separated by the slip plane, and here in the bilayer system it is the intra-layer elastic energy. The misfit energy in the bilayer system is the inter-layer energy, whose density is the $\gamma$-surface depending on the disregistry $\phi$ \cite{vitek1968intrinsic}, denoted by $\bar{\gamma}(\phi)$.  Under the assumption that $u^+(x)=-u^-(x)$, the elastic energy density can also be written based on $\phi$.  That is,
 \begin{flalign}
      	E_\text{PN}&=   	E_\text{elas}[\phi]   +	E_\text{mis}[\phi],\label{classical}\\
    	E_\text{elas}[\phi]&=\int_\sR \frac{1}{4}\bar{\beta} \left(\frac{\,\D \phi}{\,\D x}\right)^2 \,\D x,\label{classical1}\\
   	E_\text{mis}[\phi]&=\int_\sR \bar{\gamma}(\phi)\,\D x.\label{classical2}
 \end{flalign}
 Here $E_\text{PN}$ is the total energy in the classical Peierls-Nabarro model,
 $E_\text{elas}[\phi]$ and	$E_\text{mis}[\phi]$ are the elastic energy and misfit energy, respectively,
 and $\bar{\beta}$ is the elastic constant.
  The $\gamma$-surface $\bar{\gamma}(\phi)$ can be calculated from the atomistic model as the energy increment when the perfect lattice system has a uniform shift $\phi$ between the two layers \cite{vitek1968intrinsic}.

  Note that a rigorous derivation from atomistic model to the classical Peierls-Nabarro model for the dislocation in a bilayer system with the same type of atoms has been presented in Ref.~\cite{luo2018atomistic}.

   \subsection{Atomistic model of bilayer system of HEA with an edge dislocation}\label{atom_PN}

    The atomic structure of a perfect bilayer HEA is shown in Fig.~\ref{bilayer}(a). Similarly to the single layer HEA discussed in Sec.~\ref{subsec:atom_single}, we denote $j$-th atom at upper (or bottom) layer as $a_j^+$ (or $a_j^-$) and the random variable to describe the element at $a_j^+$ (or $a_j^-$) as $\omega_j^+$ (or $\omega_j^-$). Each random variable $\omega_j^+$ or $\omega_j^-$ has the distribution  in Eq.~\eqref{simple}. The averaged locations of atoms of the bilayer system form a triangular lattice.

    We consider the pairwise potential with nearest neighbor interaction.
    For the upper layer, the random lattice constant $h(\omega_j^+,\omega_{j+1}^+)$ due to the intra-layer interaction is denoted as $h^+_j$, and similarly $h(\omega_j^-,\omega_{j+1}^-)$  as $h^-_j$ in the lower layer. The distance between neighboring inter-layer atoms is $h(\omega_j^+,\omega_{j}^-)$ or $h(\omega_j^+,\omega_{j-1}^-)$.
     Following Sec.~\ref{subsec:atom_single}, the potential for the intra-layer atomic interaction is $V\left(h_j^\pm+u_{j+1}^\pm-u_j^\pm,\omega_j^\pm,\omega_{j+1}^\pm\right)$, and the inter-layer atomic interaction is denoted as $U\left(h(\omega_j^+,\omega_{j}^-),\omega_j^+,\omega_{j}^-\right)$
     or $U\left(h(\omega_j^+,\omega_{j-1}^-),\omega_j^+,\omega_{j-1}^-\right)$.
    We assume that $\rmE(h^+_j)=\rmE(h^-_j)=\rmE(h(\omega_j^+,\omega_{j}^-))=\rmE(h(\omega_j^+,\omega_{j-1}^-))=\bar{h}$. Note that the inter-layer interaction potential and the average inter-atomic distance across the two layers may be different from  those within each layer, and this does not lead to essential difference in the derivation of the continuum model.

     We consider an edge dislocation with Burgers vector $\vb=(\bar{h},0)$
     in bilayer HEA as shown in  Fig.~\ref{bilayer}(b). The displacement field at the lattice sites  $\{a^+_j,a^-_j\}_{j\in\sZ}$ is $\{u^+_j,u^-_j\}_{j\in\sZ}$, which satisfies
   \begin{equation}
   	\lim_{j\to-\infty}u^+_j-u^-_j=0,~\lim_{j\to+\infty}u^+_j-u^-_j=\bar{h}.\label{boundary}
   \end{equation}

\begin{figure}[h]
    	\centering
\includegraphics[width=1.0\textwidth]{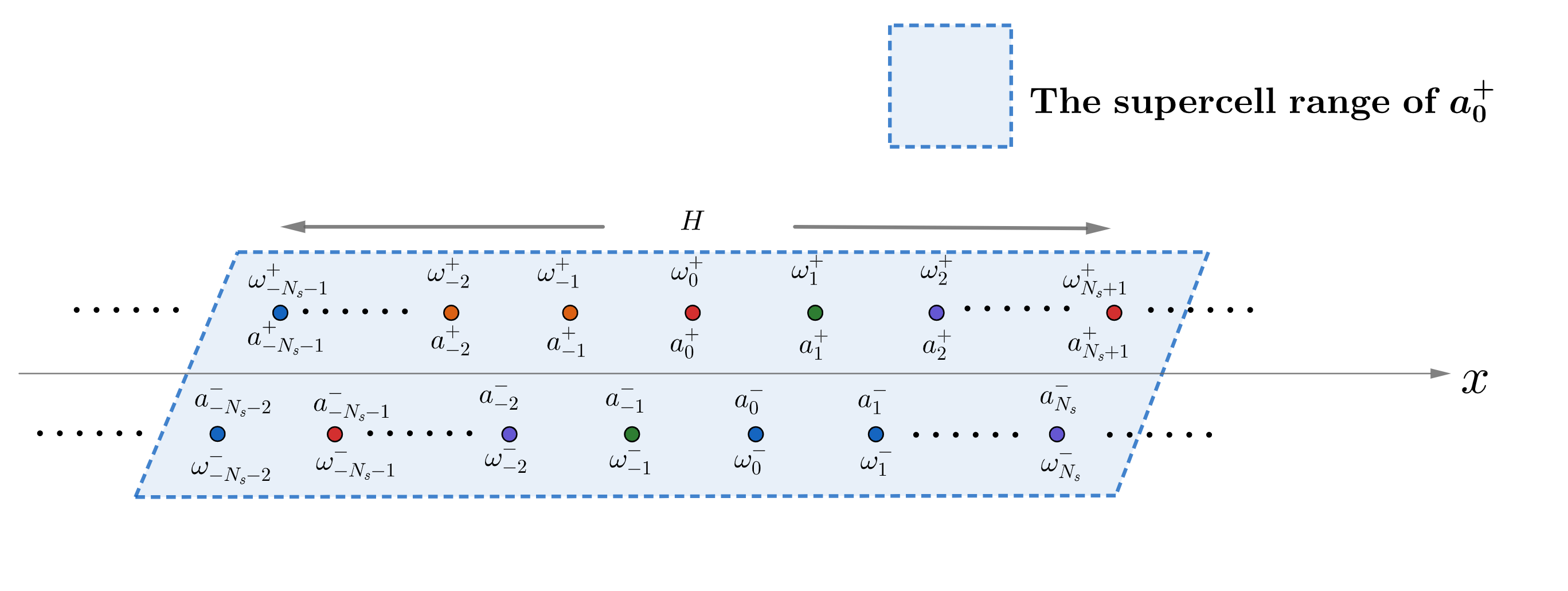}
    	\caption{Atomistic model for a bilayer HEA. (a) Perfect bilayer HEA and a supercell of size $H=(2N_s+2)\bar{h}$. (b) Bilayer HEA with an edge dislocation. The notation $\perp$ shows location of the dislocation.}
    	\label{bilayer}
    \end{figure}

In this paper, we focus on the models for HEAs with short-range order. Similarly to the $\alpha$-mixing coefficients and Assumption~\ref{assump1} for a single layer HEA in Sec.~\ref{assum elas},
we first generalize the definition of $\alpha$-mixing coefficients to two dimensions based on the bilayer HEAs, and then assume a rapidly decaying property of it.

\begin{defi}
  Denote \begin{equation}
  	\sB:=\{\omega_{j}^+\}_{j\in\sZ}\cup\{\omega_{j}^-\}_{j\in\sZ},
  \end{equation}
  and define a dimensionless distance in $\sB$ based on the atomic distance of $\{a_j^\pm\}_{j\in\sZ}$: \begin{equation}
  \left\{
  \begin{array}{l}
  B(\omega_j^\pm,\omega_{j+n}^\pm)=|n|,\vspace{1ex}\\
  B(\omega_j^-,\omega_{j+n}^+)=\sqrt{\left(n+\frac{1}{2}\right)^2+\frac{3}{4}},\vspace{1ex}\\
  B(\omega_j^+,\omega_{j+n}^-)=\sqrt{\left(n-\frac{1}{2}\right)^2+\frac{3}{4}},
  \end{array}
  \right.
\end{equation}
for all $n\in\sZ$. (Note that here the last two formulas are based on a triangle lattice in the bilayer system, and for other lattices them may be slightly different.)
For any $S_1,S_2\subset\sB$, we define
\begin{equation}
B(S_1,S_2)=\inf_{\omega_1\in S_1,\omega_{2}\in S_2}B(\omega_{1},\omega_{2}).
\end{equation}
Then for any $s>0$, the $\alpha$-mixing-type coefficients $\{\alpha_s\}_{s>0}$ of $\sB$ is defined as:
\begin{flalign}
  \alpha_s:=\sup \big\{&|P(A \cap B)-P(A) P(B)|: \forall A \in \sigma(S_1),B \in \sigma(S_2),\nonumber\\
  & S_1,S_2\in\sB, B(S_1,S_2)\le s\big\}.
\end{flalign}
\end{defi}

%
    \begin{assump}\label{assump3}
   	There is a constant number $N_s$ and a constant $C$, such that the	$\alpha$-mixing coefficients $\{\alpha_s\}_{s>0}$ of $\sB$ in the bilayer HEA satisfy
   \begin{equation}
   		\alpha_s\le C{s}^{-5} \ {\rm for} \  s\leq N_s, \ {\rm and}\  \alpha_s=0 \ {\rm for} \ s>N_s.\label{eqn:decayrate3}
   	\end{equation}
   This means that $\omega_1$ and $\omega_2$ are independent if the distance $B(\omega_{1},\omega_{2})>N_s$.
  \end{assump}

 Under this assumption, each atom is only correlated with its $N_s$ nearest neighbors on each side in its own layer and in the other layer. Assumption~\ref{assump1} holds within each layer. Hence we can still use the length of range of short-range order $H=(2N_s+2)\bar{h}$ in Eq.~\eqref{smallest range} of Definition~\ref{def1} for the bilayer HEA (illustrated in Fig.~\ref{bilayer}(a)), and accordingly, we also have Assumption~\ref{atom1} for the bilayer HEA.

   \subsection{Derivation of stochastic Peierls--Nabarro model}

Now we derive a continuum stochastic model for a dislocation in the HEA with
short range order, based on the framework of the classical Peierls-Nabarro model reviewed in Sec.~\ref{subsec:cPN}, and from the atomistic model with Assumption~\ref{assump3} in Sec.~\ref{atom_PN} and Assumption~\ref{atom1} in Sec.~\ref{assum elas}.

In the stochastic Peierls-Nabarro model, the total energy of a dislocation in the bilayer HEA
    consists of the elastic energy for the intra-layer interaction and the misfit energy for the inter-layer interaction, similar to the total energy in the classical Peierls-Nabarro model given in Eqs.~\eqref{classical}--\eqref{classical2} and with stochastic energy densities. As in the classical Peierls-Nabarro model, these energies are expressed in terms of the disregistry between the two layers $\phi(x)=u^+(x)-u^-(x)$, where $u^+(x)$ and $u^-(x)$ are displacements in the upper and lower layers, respectively. We will keep leading order stochastic effects in the elastic energy and misfit energy.

 The continuum model for the stochastic elastic energy in each layer has been obtained in Sec.~\ref{sto elas}, which is
   \begin{equation}\label{elas-energy0000}
   	\,\D E_\text{elas}^\pm=\frac{1}{2}\beta(x)\left(\frac{\,\D u^\pm}{\,\D x}\right)^2\,\D x,
   \end{equation}
where the superscript ``$+$" or ``$-$" indicates the quantities in the upper or lower layer, and the stochastic elastic constant $\beta(x)$ is the OU process governed by Eq.~\eqref{limit equation} with expression in Eq.~\eqref{solution}, which has the properties \eqref{distribution} and \eqref{cov}. In the classical Peierls--Nabarro model, it is assumed
that $u^+(x) = u^-(x)$, and accordingly, $u^+(x) = -u^-(x) = \frac{1}{2}\phi (x)$. With this condition, the total stochastic elastic energy $E_\text{elas}=E_\text{elas}^++E_\text{elas}^-$  can be written as
   \begin{equation}\label{elas-energy0001}
   	\,\D E_\text{elas}=\frac{1}{4}\beta(x)\left(\frac{\,\D \phi}{\,\D x}\right)^2\,\D x.
   \end{equation}

  Now consider the stochastic misfit energy. Following the definition~\cite{vitek1968intrinsic}, the misfit energy density, i.e., the $\gamma$-surface $\gamma(\phi)$,  is calculated from the atomistic model as the energy density increment when the perfect lattice system has a uniform shift $\phi$ between the two layers.
    Convergence from atomistic model to the $\gamma$-surface for dislocations in a system with same type of atoms (i.e., the deterministic case) has been rigourously proved in Ref.~\cite{luo2018atomistic}. In the bilayer HEA, when the average lattice is a triangular one with lattice constant $\bar{h}$, the energy density increment near atom $a_j^+$ is
     \begin{flalign}   	
     \gamma_j(\phi):=\gamma\left(\phi,\omega_j^+,\omega_{j-1}^-,\omega_{j}^-\right)=&\frac{1}{\bar{h}}\big[
     V(h_\phi(\omega_j^+,\omega_{j-1}^-),\omega_j^+,\omega_{j-1}^-)
     -V(h(\omega_j^+,\omega_{j-1}^-),\omega_j^+,\omega_{j-1}^-)\nonumber\\
&     \ \ +V(h_\phi(\omega_j^+,\omega_{j}^-),\omega_j^+,\omega_{j}^-)
     -V(h(\omega_j^+,\omega_{j}^-),\omega_j^+,\omega_{j}^-)
     \big], \label{atom-gamma}
   \end{flalign}
where $h_\phi(\omega_j^+,\omega_{j}^-)$ is the distance between the two atoms $a_j^+$ and $a_j^-$ on the two layers after a uniform shift of $\phi$ between the two layers, and same for $h_\phi(\omega_j^+,\omega_{j-1}^-)$.
That is, when the vector between the two atoms is $(h_1,h_2)$ with $\sqrt{h_1^2+h_2^2}=h$, then $h_\phi=\sqrt{(h_1+\phi)^2+h_2^2}$.


Let $\gamma\left(x,\phi\right)$ be the stochastic $\gamma$-surface in the continuum model.
That is, for the misfit energy, we have
  \begin{equation}\label{mis-energy0001}
   	\,\D E_\text{mis}=\gamma\left(x,\phi\right)\,\D x.
   \end{equation}
 We obtain  $\gamma\left(x,\phi\right)$ from the atomistic model similarly to the derivation of the stochastic elastic modulus $\beta(x)$ in Sec.~\ref{sto elas}, by averaging $\gamma$ and $\D \gamma$ over the top layer within the supercell with size $H$. (Note that averaging over the bottom layer will give the same results.)

 As in Lemmas \ref{lemma1} and \ref{convergence}, here following Assumption~\ref{assump3} for the bilayer HEA, we have that for the atomic-level $\gamma$ surface  $\{\gamma_j(\phi)=\gamma\left(\phi,\omega_j^+,\omega_{j-1}^-,\omega_{j}^-\right)\}$ defined in Eq.~\eqref{atom-gamma},
 its $\alpha$-mixing coefficients  satisfy
   \begin{equation}
   		\alpha_n\le C{n}^{-5} \ {\rm for} \  n\leq N_s+1, \ {\rm and}\  \alpha_n=0 \ {\rm for} \ n>N_s+1,
   	\end{equation}
   and this means that $\gamma_j(\phi)$ and $\gamma_{j+n}(\phi)$ are independent if $n>N_s+1$, i.e.,
\begin{equation}
\operatorname{Cov}(\gamma_j,\gamma_{j+n})=0 \ {\rm for}\  n>N_s+1 \ {\rm and \ any} \ \phi.
\end{equation}
Moreover, we have
\begin{equation}
\lim_{m\to+\infty}f_\gamma(m,\phi)=\Delta^2_\gamma(\phi),
\end{equation}
where
\begin{flalign}		f_\gamma(m,\phi):=&\frac{1}{m}\operatorname{Var}\left(\sum_{j=1}^{m}\gamma_j(\phi)\right)=\frac{1}{m}\sum_{1\le i,j\le m}\operatorname{Cov}(\gamma_j(\phi),\gamma_i(\phi)),\\
\Delta^2_\gamma(\phi):=&\sum_{j= -N_s-1}^{N_s+1}\operatorname{Cov}(\gamma_j(\phi),\gamma_0(\phi)).\label{variance_phi}
\end{flalign}
By  the
generalized central limit theorem (Theorem~\cite[Theorem 27.4]{billingsley2008probability} in Sec.~\ref{assum elas}) for $\{\gamma_j\}$,
and  Definition~\ref{def:average} in Sec.~\ref{assum elas} for the continuum limit of average of random variables $\{\gamma_j\}$ of the top layer within a supercell, as that for $\beta(x)$, we have the equation for $\gamma(x,\phi)$:
   \begin{equation}
    \frac{1}{2}\,\D \gamma(x,\phi)=\frac{\bar{\gamma}(\phi)-\gamma(x,\phi)}{H}\,\D x+\frac{\Delta_\gamma(\phi)}{\sqrt{H}}\,\D B_x,
    \label{limit equation_gamma}
   \end{equation}
where $\bar{\gamma}\left(\phi\right):=\rmE\left(\gamma_j(\phi)\right)$. That is, $\gamma(x,\phi)$ is also an OU process.

Equation \eqref{limit equation_gamma} holds pointwisely on the continuum level, and the solution is
  \begin{equation}
\gamma\left(x,\phi(x)\right)=\int_{-\infty}^x\frac{2}{H}\bar{\gamma}\left(\phi(s)\right)
e^{\frac{2}{H}(s-x)}\,\D s
+\int_{-\infty}^x\frac{2}{\sqrt{H}}
\Delta_\gamma(\phi(s))e^{\frac{2}{H}(s-x)}\,\D B_s.\label{solution misfit0}
  \end{equation}
This solution formula can be further simplified to remove the nonlocal dependence on $\phi(x)$
based on Assumption \ref{atom1}: $H\ll L$, where $L$ is the length scale of the continuum model. The simplified solution formula is
  \begin{equation}
\gamma\left(x,\phi(x)\right)=\bar{\gamma}\left(\phi(x)\right)+\Delta_\gamma(\phi(x))Y_x,
\label{solution misfit}
  \end{equation}
where $Y_x$ is the OU process given in Eq.~\eqref{ou_solution}.

The approximation of the deterministic integral in Eq.~\eqref{solution misfit0} by $\bar{\gamma}\left(\phi(x)\right)$ in Eq.~\eqref{solution misfit} is directly from the Laplace method. For the approximation of the  stochastic integral in Eq.~\eqref{solution misfit0} by $\Delta_\gamma(\phi(x))Y_x$ in Eq.~\eqref{solution misfit}, we can show formally that the relative error is small:
\begin{flalign} &\rmE\left[\int_{-\infty}^x\frac{2}{\sqrt{H}}\left[\Delta_\gamma(\phi(x))-\Delta_\gamma(\phi(s))\right]
e^{\frac{2}{H}(s-x)}\,\D B_s\right]^2\notag\\
=&\int_{-\infty}^x\frac{4}{H}\left[\Delta_\gamma(\phi(x))-\Delta_\gamma(\phi(s))\right]^2e^{\frac{4}{H}(s-x)}\,\D s\notag\\
=&\int_{-\infty}^x\left[\Delta_\gamma(\phi(x))-\Delta_\gamma(\phi(s))\right]^2\,\D e^{\frac{4}{H}(s-x)} \notag\\
=&-\int_{-\infty}^x2\left[\Delta_\gamma(\phi(s))-\Delta_\gamma(\phi(x))\right]\frac{\,\D \Delta_\gamma(\phi(s))}{\,\D s}e^{\frac{4}{H}(s-x)}\,\D s\notag\\
\ll&\int_{-\infty}^x\frac{4}{H}\left(\Delta_\gamma(\phi(s))\right)^2e^{\frac{4}{H}(s-x)}\,\D s, \ \ H\ll L \notag\\
=&\rmE\left[\int_{-\infty}^x\frac{2}{\sqrt{H}}\Delta_\gamma(\phi(s))e^{\frac{4}{H}(s-x)}\,\D B_s\right]^2.
\end{flalign}

The total stochastic energy density $W\left(x,\frac{\,\D \phi}{\,\D x},\phi\right)$
includes elastic energy and misfit energy as well as their correlation.   The total energy $E_\text{PN}=E_\text{elas}+E_\text{mis}$ can be written as
\begin{equation}\label{PN-energy0001}
   	\,\D E_\text{PN}=W\left(x,\frac{\,\D \phi}{\,\D x},\phi\right)\,\D x.
   \end{equation}
The stochastic elastic energy density $\frac{\D E_\text{elas}}{\D x}$ has been obtained in
Eq.~\eqref{elas-energy0001} and stochastic $\gamma$-surface $\gamma\left(x,\phi\right)$ in
Eq.~\eqref{solution misfit}. Following the same argument for the continuum limit of the average over the top layer within a supercell,  $W\left(x,\frac{\,\D \phi}{\,\D x},\phi\right)$ is also an OU process that satisfies
   \begin{equation}
   	\frac{1}{2}\D W\left(x,\frac{\,\D \phi}{\,\D x},\phi\right)=\frac{\bar{W}\left(\frac{\,\D \phi}{\,\D x},\phi\right)-W\left(x,\frac{\,\D \phi}{\,\D x},\phi\right)}{H}\,\D x+\frac{1}{\sqrt{H}}\Delta_W\left(\frac{\,\D \phi}{\,\D x},\phi\right)\,\D B_x,\label{limit pn}
   \end{equation}
where
\begin{equation}
\Delta_W^2\left(\frac{\,\D \phi}{\,\D x},\phi\right):=\sum_{j= -N_s-1}^{N_s+1}\operatorname{Cov}\Big(W_j({\textstyle \frac{\,\D \phi}{\,\D x}},\phi),W_0({\textstyle \frac{\,\D \phi}{\,\D x}},\phi)\Big),
\end{equation}
\begin{flalign}W_j\left(\frac{\,\D \phi}{\,\D x},\phi\right):=W\left(\frac{\,\D \phi}{\,\D x},\phi,\omega_j^+,\omega_{j+1}^+,\omega_{j-1}^-,\omega_{j}^-\right)
:=\frac{1}{8}(\beta^+_j+\beta^-_{j-1})\left(\frac{\,\D \phi}{\,\D x}\right)^2+\gamma_j(\phi),
\label{energy density}
   \end{flalign}
   and
   \begin{equation}
 \bar{W}\left(\frac{\,\D \phi}{\,\D x},\phi\right)  =\rmE W_j\left(\frac{\,\D \phi}{\,\D x},\phi\right) .
   \end{equation}
The solution of Eq.~\eqref{limit pn}, to the leading order as that in Eq.~\eqref{solution misfit}, is
  	 \begin{equation}
  		W\left(x,\frac{\,\D \phi}{\,\D x},\phi\right)=\bar{W}\left(\frac{\,\D \phi}{\,\D x},\phi\right)
  +\Delta_W\left(\frac{\,\D \phi}{\,\D x},\phi\right)Y_x,\label{solution pn}
  	\end{equation}
where $Y_x$ is the OU process given in Eq.~\eqref{ou_solution}.

Based on  Eqs.~\eqref{solution}, \eqref{solution misfit} and \eqref{solution pn}, the stochastic total energy density can be written as
\begin{flalign}
	W\left(x,\frac{\,\D \phi}{\,\D x},\phi\right)=&\frac{1}{4}\bar{\beta}\left(\frac{\,\D \phi}{\,\D x}\right)^2+\bar{\gamma}(\phi)+\Delta_W\left(\frac{\,\D \phi}{\,\D x},\phi\right)Y_x,\label{solution pn1}
\end{flalign}
with OU process $Y_x$ given in Eq.~\eqref{ou_solution}, and
\begin{flalign}
\Delta_W\left(\frac{\,\D \phi}{\,\D x},\phi\right)=&\frac{1}{4}{\Delta_\text{e}}\left(\frac{\,\D \phi}{\,\D x}\right)^2+\Delta_\gamma(\phi)+\Delta_C\left(\frac{\,\D \phi}{\,\D x},\phi\right).\label{solution pn2}
\end{flalign}
Here in Eq.~\eqref{solution pn1}, the first two terms  give the average total energy density (corresponding to that in the classical Peierls-Nabarro model), and the integral term gives  the stochastic effects in the total energy, which include contributions from the elastic energy (the $\Delta_\text{e}$ term), the misfit energy (the $\Delta_\gamma$ term) and their correlation (the $\Delta_C$ term) as given in Eq.~\eqref{solution pn2}. Note that the correlation between the elastic energy density and the $\gamma$-surface is characterized by
$\Delta_C\left(\frac{\,\D \phi}{\,\D x},\phi\right)=\Delta_W\left(\frac{\,\D \phi}{\,\D x},\phi\right)-\frac{1}{4}{\Delta_\text{e}}\left(\frac{\,\D \phi}{\,\D x}\right)^2-\Delta_\gamma(\phi)$.

If there is no short-range order, the length of range of short-range order $H\to 0$. In this case,
Eq.~\eqref{indpendent elas}  holds for the elastic energy, and from Eqs.~\eqref{limit equation_gamma}  and \eqref{limit pn} with $H\to 0$,
 we have
 \begin{flalign}
\gamma(x,\phi)\,\D x=&\bar{\gamma}(\phi)\,\D x+\Delta_\gamma(\phi)\sqrt{\bar{h}}\,\D B_x,\\
W\left(x,\frac{\,\D \phi}{\,\D x},\phi\right)\,\D x=&\left[\frac{1}{4}\bar{\beta}\left(\frac{\,\D \phi}{\,\D x}\right)^2+\bar{\gamma}(\phi) \right]\,\D x+\Delta_W\left(\frac{\,\D \phi}{\,\D x},\phi\right)\sqrt{\bar{h}}\,\D B_x.
\end{flalign}
 Recall that $\bar{h}$ is the average lattice constant in the HEA which is the smallest distance between atomic sites.
   This agrees with the energy formulation derived in Ref.~\cite{jiang2020stochastic}  under the assumption of independent randomness, i.e., without short-range order.

  \subsection{Equation of Stochastic Peierls--Nabarro Model}\label{equation pn}

  We have obtained a stochastic energy whose density is given in Eq.~\eqref{solution pn1}.
 This is different from the stochastic PDEs studied in the literature in which a stochastic term in the form of white noise is directly added in a deterministic PDE (e.g. \cite{walsh1986introduction,KS91,allen1998finite,du2002numerical,E2010,Lord2014,Feng2017,cao2018finite}).  We briefly discuss the variational formulation of the obtained stochastic energy in this section.
For simplicity, we start from the case without short-range order and the stochastic energy depends only on $\phi$ (i.e., coming only from the misfit energy as discussed in Ref.~\cite{zhang2019effect}).
In this case, the total energy of the bilayer HEA is:
\begin{equation}
  	E_\text{PN}[\phi]=\int_{\sR}\left[\frac{1}{4}\bar{\beta}\left(\frac{\,\D \phi}{\,\D x}\right)^2+\bar{\gamma}(\phi)\right]\,\D x+\int_{\sR}\sigma\left(\phi\right)\,\D B_x.\label{energy}
  \end{equation}
For an equilibrium state of this stochastic energy,  using  the Euler-Lagrange equation formulation $\frac{\delta E_\text{PN}}{\delta \phi}=0$ formally, we have
\begin{equation}
 -\frac{1}{2}\beta\left(\frac{\,\D^2 \phi}{\,\D x^2}\right)+\gamma'(\phi)+\sigma'(\phi)\dot{W}=0, \label{s-equation}
  \end{equation}
  where $\dot{W}$ is the derivative of Brown motion in space (i.e., the Gaussian white noise in space). This means that for any Brownian motion path,  the energy $E_\text{PN}[\phi]$ in \eqref{energy} is in equilibrium.
  Similarly, for the dynamics problem,  the stochastic energy in \eqref{energy} formally leads to the following gradient flow equation $\frac{\D\phi}{\D t}=-M \frac{\delta E_\text{PN}}{\delta \phi}$:
\begin{equation}
 \frac{\D\phi}{\D t}=M\left(	\frac{1}{2}\beta\left(\frac{\,\D^2 \phi}{\,\D x^2}\right)-\gamma'(\phi)-\sigma'(\phi)\dot{W}\right), \label{s-gradient flow}
  \end{equation}
where $M>0$ is the mobility.

When the short-range order is considered, in the case where stochastic energy depends only on $\phi$ (i.e., coming only from the misfit energy as discussed in Ref.~\cite{zhang2019effect}) as discussed above, the total energy of the bilayer HEA is:
\begin{equation}
  	E_\text{PN}[\phi]=\int_{\sR}\left[\frac{1}{4}\bar{\beta}\left(\frac{\,\D \phi}{\,\D x}\right)^2+\bar{\gamma}(\phi)+\sigma\left(\phi\right)Y_x\right]\D x,\label{energy_sro}
  \end{equation}
where $Y_x$ is the OU process defined in Eq.~\eqref{ou_solution}. Equilibrium of this stochastic energy is described by
\begin{equation}
  \frac{\delta E_\text{PN}}{\delta \phi}=	-\frac{1}{2}\beta\left(\frac{\,\D^2 \phi}{\,\D x^2}\right)+\gamma'(\phi)+\sigma'(\phi)Y_x=0, \label{s-equation_sro}
  \end{equation}
  and the gradient flow associated with it is
\begin{equation}
 \frac{\D\phi}{\D t}=M\left(	\frac{1}{2}\beta\left(\frac{\,\D^2 \phi}{\,\D x^2}\right)-\gamma'(\phi)-\sigma'(\phi)Y_x\right), \label{s-gradient flow_sro}
  \end{equation}
where $M>0$ is the mobility. We would like to remark that existence and uniqueness of the (mild) solution of the stochastic Peierls-Nabarro model in Eq.~\eqref{s-equation} or \eqref{s-equation_sro}   can be proved similarly to the results in Ref.~\cite{cao2018finite}.

 Rigorous definitions of the solutions of these stochastic equations and analysis of their properties will be explored in the future work.


%

\section{Summary}
  We have derived stochastic continuum models from atomistic models for HEAs incorporating the atomic level randomness and short-range order, for both the elasticity in HEAs without defects and HEAs with dislocations. The stochastic continuum model for dislocations in HEAs is under the framework of Peierls-Nabarro-type  models which are able to include the dislocation core effect. The obtained stochastic continuum descriptions for the  atomic level randomness with short-range order are in the form of  OU processes, which validates the continuum model adopted phenomenologically in the stochastic Peierls-Nabarro model for dislocations in HEAs proposed in \cite{zhang2019effect}.

  A critical quantity in the continuum limit from the atomistic model is the characteristic length $H$ of the short-range order on the atomistic level, and this characteristic length is kept in the continuum limit process. When $H$  goes to $0$, the stochastic continuum models obtained in this paper  recover the continuum models for HEAs without short-range order proposed and analyzed previously \cite{zhang2019effect,jiang2020stochastic}.
Moreover, in the continuum limit from the atomistic model, we keep both the atomic level mean and variance when averaging is performed.


The obtained stochastic continuum models are based on the energy formulation. We also briefly discuss the variational formulation, i.e., the associated stochastic equations, of these obtained stochastic energies.


The stochastic continuum models for elasticity and dislocations in HEAs can be generalized to the settings of two or three dimensions \cite{Xu2000,Shen-Wang2004,xiang2008generalized,wei2008generalized}, which will be explored in the future work.  Analysis of  the obtained stochastic equations and their numerical solutions will also be considered in the future work.

  \section*{Acknowledgement}
  This work was supported by the Hong Kong Research Grants Council General Research Fund 16307319, and the Project of Hetao Shenzhen-HKUST Innovation Cooperation Zone HZQB-KCZYB-2020083.

  \bibliographystyle{unsrtnat}
  \bibliography{references}

\end{document}